\RequirePackage{fix-cm}
\documentclass[smallextended]{svjour3}       
\smartqed  
\usepackage{graphicx}
\usepackage{amsfonts}
\usepackage{amsmath}
 \usepackage{mathptmx}      
\usepackage{latexsym}
\begin{document}

\title{Maximization of Relative Social Welfare on Truthful Cardinal Voting Schemes
}


\author{Sinya Lee
}


\institute{Anthony Sinya Lee\at
              Institute for Interdisciplinary Information Sciences, Tsinghua University, Beijing 100084, China \\
              \email{sinya8@gmail.com}    
}

\date{Dec 11, 2013 / Revised: Mar 31, 2019}

\maketitle

\begin{abstract}
Consider the the problem of maximizing the relative social welfare of truthful single-winner voting schemes with cardinal preferences compared to the classical range voting scheme. The range voting scheme is a simple and straightforward mechanism which deterministically maximizes the social welfare. However, the scheme that is known to be non-truthful and we studied the truthful mechanism that maximize the ratio of its expected social welfare to the social welfare achieved by the range voting scheme. We provide a scheme which achieve a ratio of $\Omega(m^{-2/3})$ in this paper. It is proved that this bound is tight asymptotically and it is impossible to find a better voting scheme.
\end{abstract}

\section{Introduction}
\label{intro}
A single-winner voting scheme is a function that takes the preference profile of $n$ voters on $m$ candidates as input and generates a single winner. Single-winner voting schemes are widely used in real life from the election of the US president to the decision of party date in a club. Traditionally scientists focus on the setting that the preference of each voter is a total order on the candidates. The classical result by Gibbard \cite{Gib1973} and and Satterthwaite  \cite{Sat1975} shows that under such setting there doesn't exist any non-dictatorial and truthful scheme with at least 3 candidates possible to win the election. Scientists are working on the problem of manipulation of such voting schemes such as the recent papers \cite{Fri2008} and \cite{Xia2008}.

In this paper we study the voting problem under cardinal preference setting. That is, for each voter $i \in N$, its preference is a function $u_i: M \to \mathbb [b_1, b_2]$ where $N = \{1,\dots, n\}$ is the voter set, $M = \{1,\dots, m\}$ is the candidate set and $b_1 < b_2$ are the lower and upper bounds of possible ``scores'' for the candidates. Hence the preference profile $u$ can be defined as $\mathbf u = (u_1,\dots, u_n)$. Under such cardinal preference setting, we are able to define the cardinal social welfare of choosing candidate $j$ and  $\operatorname{Wel}(j, \mathbf u) = \sum_{i=1}^n u_i(j)$.

A simple and straightforward mechanism is the \textit{range voting} scheme $RV$. The range voting scheme is the scheme that the winner of the election is the candidate $j$ with maximized $\operatorname{Wel}(j, \mathbf u)$. That is, the range voting scheme always returns the candidate with highest total score.

It is obvious that the range voting scheme maximizes the social welfare. However, the range voting is not truthful. It fails as the way that the general ordinal scheme fails. When the top 2 candidates has a tie (or tiny difference), it is possible and motivated for some voters to manipulate the result if they prefer the second top candidate to the top candidate. Thus the range voting scheme is not truthful.

Therefore, it is natural to study how much social welfare we can get by a truthful mechanism compared to the non-truthful range voting scheme.

Since the preference function $u_i$ is always well-defined and preserve the preference ordering of the candidates under positive affine transformation, we may regard $x \to au_i(x) + b$ as another description of $x \to u_i(x)$ for any $a>0$ and $b$. Therefore we can apply an different positive affine transformation to each voter's preference function and we can make the following assumption.

\spnewtheorem{assumption}{Assumption}{\bf}{\it}

\begin{assumption}[Preference Assumption]
\label{ass}
The preference function $u_i$ of each voter $i \in N$ is a function from $M$ to $[0,1]$ such that there exists some $j_0$ and $j_1$ in $M$ such that $u_i(j_0) = 0$ and $u_i(j_1) = 1$.
\end{assumption}

With this assumption, for any $i\in N$ and $j \in M$, $u_i(j)\ge 0$. Also, there exists some voter $i^* \in N$ and candidate $j^* \in M$ such that $i$'s preferences on $j$, $u_{i^*}(j^*) > 0$. Therefore we have the social welfare $\operatorname{Wel}(RV(\mathbf u), \mathbf u) \ge \operatorname{Wel}(j^*, \mathbf u) > 0$. Hence it is possible and reasonable to define the relative social welfare or welfare ratio for any mechanism $J$ as
$$
\operatorname{ratio}(J) = \inf_{\mathbf u} \frac{\mathbf E[\operatorname{Wel}(J(\mathbf u), \mathbf u)]}{\operatorname{Wel}(RV(\mathbf u), \mathbf u)}
$$

Here we show that this assumption is sufficient to make the ratio function well defined. Furthermore, we will show that this assumption is necessary to make the problem interesting enough to study, which we be further discussed in Section \ref{dis}.

A less general version of this problem is raised in Filos-Ratsikas and Miltersen's recent paper \cite{Rat2013}, except they have one more unnecessary assumption requiring that there can't be any ties in each voter $i$'s preference function. In their paper they provide a truthful mechanism with a welfare ratio of $\Omega(m^{-3/4})$. They also provide a negative result for ordinal mechanisms (mechanisms that depends only on the candidates' ranking of each voter) stating that any ordinal truthful mechanism can only achieve a welfare ratio of $O(m^{-2/3})$.

In this paper, we improve the result in \cite{Rat2013}. We will show a truthful mechanism with welfare ratio $\Omega(m^{-2/3})$, much better than the result in \cite{Rat2013}. The formal statement and proof of this result will be in Section \ref{res1}. This mechanism cannot be further improved asymptotically due to the upper bound proved in \cite{Rat2013}. In additional, the result above require the assumption requiring all preference functions to have no tie. Therefore my results are more general and universal than the results in \cite{Rat2013}.

\section{Preliminaries and Notation}
\label{pre}

\subsection{Voting Systems}
\label{preVote}

Let $N = \{1,\dots, n\}$ denote the set of voters or agents. Let $M = \{1,\dots, m\}$ denote the set of candidates or alternatives.

A preference is a function $u : M \to [0,1]$ such that $u(j_0) = 0$ and $u(j_1) = 1$ for some $j_0$ and $j_i$ in $M$. Let $V$ denote the set of all possible preference functions.

A preference profile $\mathbf u \in V^n$ is a $n$-tuple of preference functions $(u_1,\dots, u_n)$. Also for any preference profile $\mathbf u = (u_1,\dots, u_n)$ and any candidate $j \in M$, we define $\mathbf u(j) = \sum_{i=1}^n u_i(j)$ for convenience.

A voting scheme or mechanism is a randomized function $J: V^n \to M$. That is, a randomized function that generate a single winner from a preference profile. We can also regard a voting scheme $J$ as a function from $V^n$ to a distribution over $M$, say, $\Delta(M)$. Moreover, for a mechanism $J$ that already known to be deterministic, we can simply regard $J$ as a deterministic function from $V^n$ to $M$. Throughout this paper we may use $J$ under any one of the three different definitions above in the most convenient way. Since the exact definition of $J$ will be clear in the context, we don't need to worry about the ambiguity. Let $\mathbf{Mech}$ denote the set of all possible voting schemes.

\subsection{Classification of Voting Schemes}
\label{preCla}

We say a voting scheme $J \in \mathbf{Mech}$ to be truthful or strategy-proof if for any preference profile $\mathbf u = (u_1,\dots,u_n) \in V^n$, any voter $i$ and any preference function $u_i'$, we have $u_i(J(u_i, u_{-i})) \ge u_i(J(u_i', u_{-i}))$. Let $\mathbf{Mech^T}$ denote the set of all truthful mechanisms in $\mathbf{Mech}$.

For any two preference functions $u, u' \in V$, we say $u$ and $u'$ are ordinal equivalent (denoted by $u \sim u'$) if for any $j_0, j_1 \in M$, $u(j_0) > u(j_1)$ if and only if $u'(j_0) > u'(j_1)$. That is, whether we order the every candidate $j \in M$ by $u(j)$ or by $u'(j)$, the ordering will be the same (and the ties will also be preserved moreover).

For any $i \in N$ and two preference profile $\mathbf u, \mathbf u' \in V^n$ where $\mathbf u = (u_1, \dots, u_n)$ and $\mathbf u' = (u'_1,\dots, u'_n)$, we say $\mathbf u$ and $\mathbf u'$ $i$-ordinal equivalent (dented by $\mathbf u \overset{i}{\sim} \mathbf u'$) if $u_i \sim u_i'$ and $u_{j} = u_j'$ for every $j \ne i$. We say $\mathbf u$ and $\mathbf u'$ ordinal equivalent (denoted by $\mathbf u \sim \mathbf u'$) if $\bigwedge_{i=1}^m (u_i \sim u_i')$.

For a voting scheme $J \in \mathbf{Mech}$ and any $i \in N$, we say $J$ is $i$-ordinal if $J(\mathbf u) = J(\mathbf u')$ for any $\mathbf u, \mathbf u' \in V^n$ with $\mathbf u \overset{i}{\sim} \mathbf u'$. For a voting scheme $J \in \mathbf{Mech}$, we say $J$ is ordinal if $J(\mathbf u) = J(\mathbf u')$ for any $\mathbf u, \mathbf u' \in V^n$ with $\mathbf u \sim \mathbf u'$. Let $\mathbf{Mech}^{i\textbf{-O}}$ denote the set of all $i$-ordinal mechanisms in $\mathbf{Mech}$ for any $i \in N$. Similarly, let $\mathbf{Mech^O}$ be the set of all ordinal mechanisms in $\mathbf{Mech}$.

We say a voting scheme $J \in \mathbf{Mech}$ is unilateral if there exists some $i \in N$ such that, $J(u_i, u_{-i}) = J(u_i, u_{-i}')$ for any $u_i \in V$ and $u_{-i}, u_{-i}' \in V^{n-1}$. Let $\mathbf{Mech^{1}} \subseteq \mathbf{Mech}$ denote the set of all unilateral mechanisms.

We say a voting scheme $J \in \mathbf{Mech}$ is duple if $J \in \mathbf{Mech^O}$ and there exists some different $i_0, i_1 \in N$ such that $J(u_{i_0}, u_{i_1}, u_{-i_0, -i_1}') = J(u_{i_0}, u_{i_1}, u_{-i_0, -i_1}')$ for any $u_{i_0}, u_{i_1} \in V$ and $u_{-i_0, -i_1}, u_{-i_0, -i_1}' \in V^{n-2}$. Let $\mathbf{Mech^{2}} \subseteq \mathbf{Mech^O}$ denote the set of all unilateral mechanisms.

We say a scheme $J \in \mathbf {Mech}$ is neutral if the scheme is fair to the candidates, i.e., the names of the candidates don't matter. Formally, $J$ is neutral if for any permutation $\sigma$ over $M$ and any $\mathbf u = (u_1, \dots, u_n) \in V^n$, we have $\sigma(J(\mathbf u)) = J(u_1 \circ \sigma, \dots, u_1 \circ \sigma)$. Let $\mathbf {Mech^N}$ be the set of all neutral mechanisms.

Similarly we say a scheme $J \in \mathbf{Mech}$ is anonymous if it is fair to the voters, i.e., the names of the voters don't matter. Formally, $J$ is anonymous if for any permutation $\sigma$ over $N$ and any $\mathbf u = (u_1, \dots, u_n) \in V^n$, we have $J(\mathbf u) = J(u_{\sigma(1)},\dots, u_{\sigma(n)})$. Let $\mathbf{Mech}^A$ be the set of all anonymous mechanisms.

\subsection{With and Without Ties}

Define $U \subseteq V$  be the set of preference functions without ties. Or more formally, $U = \{u \in V: \textrm{$u(j) \ne u(j')$ for any $j \ne j'$}\}$.

Therefore $U^n$ are the set of preference profiles without ties. Then we can define $\mathbf {Mech_{U}}$ to be the set of every voting scheme $J$ which takes input in $U^n$. Hence we can define $\mathbf{Mech_U^T}, \mathbf {Mech}_\mathbf{U}^{i\textbf{-O}}, \mathbf{Mech_U^O}, \mathbf{Mech_U^1}, \mathbf{Mech_U^2}, \mathbf{Mech_U^N}, \mathbf{Mech_U^A} \subseteq \mathbf{Mech_U}$ be the sets of truthful, $i$-ordinal, ordinal, unilateral, duple, neutral and anonymous mechanisms in $\mathbf{Mech_U}$ in the obvious way.

Notice that for any $J \in \mathbf {Mech}$, we can induce $J_U \in \mathbf{Mech_U}$ from $J$ naturally by setting $J_U(\mathbf u) = J(u)$ for any $\mathbf u \in U^n$. And then we have the following lemmas:

\begin{lemma}
\label{lemU}
If $J \in \mathbf {Mech^C}$, then $J_U \in \mathbf {Mech^C_U}$, where $\mathbf C$ can be $\mathbf{T}, i\textbf-\mathbf O, \mathbf O, 1, 2, \mathbf N, \mathbf A$.
\end{lemma}

\begin{proof}
Just by plugging in that $J_U(\mathbf u) = J(\mathbf u)$ for every $u \in U^n$ and checking the definition of each $\mathbf {Mech^C_U}$, we can prove the correctness of the five lemmas above.
\end{proof}

\subsection{Some Mechanisms}
\label{preMec}

Define $RV \in \mathbf{Mech}$ to be the range voting scheme. That is, the voting scheme such that the winner $RV(\mathbf u)$ is always the candidate $j^*$ with maximized $\operatorname{Wel}(j^*, u)$

For any $q \in \{1,\dots, m\}$, let $J^{1,q} \in \mathbf{Mech^T}$ be the mechanism that choose a voter $i$ uniformly at random, and then choose a winner uniformly at random from voter $i$'s $q$ most favorited candidates. Here we break ties by the name of the candidates. That is, if some voter $i \in N$ equally like candidate $j_0, j_1 \in M$ with $j_0 < j_1$, then we just assume that voter $i$ likes $j_0$ better.

For any $q \in \{\lfloor n/2\rfloor + 1,\dots, n+1\}$, let $J^{2,q} \in \mathbf{Mech^T}$ be the mechanism that choose two different candidates uniformly at random, then let every voter vote for on of the two candidate he likes better. If there exists one candidate gets at least $q$ votes, let he/she be the winner. Otherwise flip a coin and pick a random winner from the two candidates we chosen. Similarly we break ties by name of the candidate.

Let random-favorite mechanism be the nickname of $J^{1,1}$.

Observe that even though $J^{1,q}$'s and $J^{2,q}$'s are not neutral, they're no-tie versions $J^{1,q}_U$'s and $J^{2,q}_U$'s are anonymous as well as neutral.

Define $J^* = \frac{1}{2} J^{1,1} + \frac{1}{2} J^{1, \lfloor m^{1/3}\rfloor } \in \mathbf {Mech^T}$. We will later show that $J^*$ is a mechanism with $\operatorname{ratio}(J)  = \Omega(m^{-2/3})$.

\subsection{Social Welfare}
\label{preSoc}

For any preference profile $\mathbf u = (u_1,\dots, u_n) \in V^n$ and any winner $j \in M$, define the social welfare $$\operatorname{Wel}(j, \mathbf u) = \sum_{i=1}^{n} u_i(j)$$.

For any mechanism $J \in \mathbf {Mech}$ and preference profile $\mathbf u \in V^n$, define the ratio of $J$ on input $\mathbf u$ to be $$\operatorname{ratio}(J, \mathbf u) = \frac{\mathbf E[\operatorname{Wel}(J(u), u)]}{\operatorname{Wel}(RV(u), u)}$$

For any mechanism $J \in \mathbf {Mech}$, we define the ratio or the relative social welfare of $J$ to be 
$$\operatorname{ratio}(J) = \inf_{\mathbf u \in V^n} \operatorname{ratio}(J, \mathbf u)$$

Similarly, for $J \in \mathbf{Mech_U}$, we can define 
$$
\operatorname{ratio_U}(J) =\inf_{\mathbf u \in U^n}  \frac{\mathbf E[\operatorname{Wel}(J(u), u)]}{\operatorname{Wel}(RV_U(u), u)}
$$

\subsection{Negative Result for Ordinal Mechanisms without Ties}

In this subsection I will show the proof of negative result for truthful ordinal mechanism where ties are not allowed, stating that all such mechanisms can only achieve a ratio of $O(m^{-2/3})$. This original proof of this statement can be find in \cite{Rat2013}.

\begin{theorem}\cite{Gib1977}
\label{Gib2}
For any $J \in \mathbf{Mech_U^O} \cap \mathbf{Mech_U^T}$, $J$ is a convex combination of mechanisms in $\mathbf{Mech_U^1} \cap \mathbf{Mech_U^O} \cap \mathbf{Mech_U^T}$ and mechanisms in $\mathbf{Mech_U^2} \cap \mathbf{Mech_U^T}$.
\end{theorem}

The proof of Theorem \ref{Gib2} is far too long to fit in a short paper like this. You can find the proof in Allan Gibbard's classical paper \cite{Gib1977} published long long ago.

\begin{theorem} \cite{Rat2013}
\label{Bar1}
For any $J \in \mathbf{Mech_U^A} \cap \mathbf{Mech_U^N} \cap \mathbf{Mech_U^O} \cap \mathbf{Mech_U^1} \cap \mathbf{Mech_U^T}$, $J$ is a convex combination of some mechanisms $J_U^{1,q}$ for $q \in \{1,\dots, m\}$. For any $J \in \mathbf{Mech_U^A} \cap \mathbf{Mech_U^N} \cap \mathbf{Mech_U^O} \cap \mathbf{Mech_U^T}$ such that $J$ can be obtained by convex combinations of duple mechanisms, we have $J$ is a convex combination of mechanisms $J_U^{2,q}$ for $q \in \{\lfloor n/2\rfloor+1, \dots, n+1\}$.
\label{Rat1}
\end{theorem}

A outline of proof of Theorem \ref{Rat1} can be found in paper \cite{Rat2013}, which is very similar to a detailed proof to a closely related statement in two papers \cite{Bar1978} and \cite{Bar1979} written in the 1970s. We won't show the proof here either.

\begin{corollary} \cite{Rat2013}
\label{cor}
For any $J \in \mathbf{Mech_U^O} \cap \mathbf{Mech_U^N} \cap \mathbf{Mech_U^A} \cap \mathbf{Mech_U^T}$, $J$ is a convex combination of $J_U^{1,q}$'s for $q \in \{1,\dots, m\}$ and $J_U^{2,q}$'s for $q \in \{\lfloor n/2\rfloor+1,n+1\}$.
\end{corollary}

\begin{proof} 
We can get the corollary directly by combining the Theorem \ref{Gib2} and Theorem \ref{Bar1}.
\end{proof}

\begin{lemma} \cite{Rat2013}
\label{neuAno}
For any $J \in \mathbf{Mech_U^O} \cap \mathbf{Mech_U^T}$, there exists $J' \in \mathbf{Mech_U^O} \cap \mathbf{Mech_U^N} \cap \mathbf{Mech_U^A} \cap \mathbf{Mech_U^T}$ such that $\operatorname{ratio_U}(J') \le \operatorname{ratio_U}(J)$.
\end{lemma}

\begin{proof}

Consider any fixed $J \in \mathbf{Mech_U^O} \cap \mathbf{Mech_U^T}$. For any permutation $\sigma$ over $N$ and $\tau$ over $M$, define $J_{\sigma,\tau}$ by $J_{\sigma, \tau}(u_1,\dots, u_n) = J(u_{\sigma(1)} \circ \tau, \dots, u_{\sigma(n)} \circ \tau)$. By this definition, for any $\mathbf u \in U^n$, there exists $\mathbf u' \in U^n$ such that $J(\mathbf u') \le J_{\sigma, \tau}(\mathbf u)$. Thus $\operatorname{ratio_U}(J_{\sigma, \tau}) \ge \operatorname{ratio_U}(J)$.

Now let $J'$ be the mechanism that choose a permutation $\sigma$ over $N$ and a permutation $\tau$ over $M$ uniformly at random. And then apply $J_{\sigma, \tau}$. More formally, $J' = \frac{1}{n!m!}\sum_{\sigma, \tau} J_{\sigma, \tau}$. Since $\operatorname{ratio_U}(J_{\sigma, \tau}) \ge \operatorname{ratio_U}(J)$ for any $\sigma$ and $\tau$, and $\operatorname{ratio_U}(J') \ge \operatorname{ratio_U}(J)$

\end{proof}

Now we are able to derive the negative result given in \cite{Rat2013}.

\begin{theorem}\cite{Rat2013}
\label{oldNeg}
For a sufficient large $m$ and a sufficiently larger $n$, for any mechanism $J \in \mathbf{Mech^O_U} \cap \mathbf{Mech^T_U}$ we have $\operatorname{ratio_U}(J) = O(m^{-2/3})$.
\end{theorem}

My proof of this theorem is based on the proof given in \cite{Rat2013}. But my proof is simplified for the case that $n$ is much larger than $m$. You can find a more general but complex proof in \cite{Rat2013}.

\begin{proof}

First by applying Lemma \ref{neuAno}, it is sufficient to consider mechanism $J$ inside $\mathbf{Mech^O_U} \cap \mathbf{Mech^T_U} \cap \mathbf{Mech^N_U} \cap \mathbf{Mech^A_U}$.

Let $k = \lfloor m^{1/3}\rfloor$ and $g = \lfloor m^{2/3} \rfloor$. Let $n = m-1+g$. Now we will try to construct a single profile $\mathbf u\in U^n$ such that $\mathbf u $ makes $\operatorname{ratio}(J, \mathbf u)$ be small for any $J$.

By applying Corollary \ref{cor}, $J$ can be seen as convex combinations of some $J_U^{1,q}$ and $J_U^{2,q}$. Then it is sufficient to show that $\operatorname{ratio}(J^{1,q}, \mathbf u)$ and $\operatorname{ratio}(J^{2,q}, \mathbf u)$ is small enough for all possible $q$.

Let $M_1,\dots,M_g$ be a partition of $\{1,\dots, kg\}$. We construct $\mathbf u = (u_1,\dots, u_n)$ as follow:

Type (1) voters: For each $i \in \{1,\dots, m-1\}$, let $u_i(i) = 1$, $u_i(m) = 0$ and let $u_i(j) < 1/m^2$ for all $j \notin\{i,m\}$. 

Type (2) voters: For each $i \in \{m,\dots, m-1+g \}$, let $u_{i}(j) > 1-1/m^2$ if $j \in M_{i-m+1}$; let $u_{i}(j) = 1-1/m^2$ if $j = m$ and let $u_i(j) < 1/m^2$ otherwise.

Note that the best social welfare $\operatorname{Wel}(m, \mathbf u)$ is $(1-1/m^2)g$ while the other candidates are all $< 2+1/m$. Thus the conditional ratio when the winner is not $m$ is $\le (2+1/m)/(1-1/m^2)g \le 3m^{-2/3} = O(m^{-2/3})$. Now all we need is to find out the probability that $m$ is chosen as the winner.

For a mechanism $J^{2,q}_U$, $m$ is chosen as the winner with probability at most $2/m$.

Now consider a mechanism $J^{1,q}_U$. If $q \le k$, $m$ can never be the winner. Otherwise $m$ is the winner if and only if we pick a voter $i$ with $u_i(m) > 0$ and choose candidate $m$ as the final winner, which has probability 
$$
\frac{g}{m-1+g} \cdot \frac{1}{q} = O(m^{-1/3})\cdot O(m^{-1/3}) = O(m^{-2/3})
$$

Now we have proved the theorem with $n = m-1+g$. When $n > m-1+g$, simply set $\mathbf u = (u_1, \dots, u_n)$ the same way as the case when $n = m-1+g$ except we let $u_i = u_{i-(m-1+g)}$ when $i > m-1+g$. That is, when the number of voters are larger than $m-1+g$, we just repeat the first $m-1+g$ voters.

Since the fraction of Type (1) voters and Type (2) voters only varies within a constant factor, we can use the same analysis method to analyze $\operatorname{ratio}(J, \mathbf u)$. Similarly expected ratio conditioned that $m$ is not chosen is $O(m^{-2/3})$ and the probability of choosing $m$ is also $O(m^{-2/3})$. Then we have proved the theorem for all $n \ge m-1+g$.

\end{proof}

\subsection{Classifications of Preference Functions}
\label{preCla}

For any $u \in U$, let $\mathcal I(u)$ denote the image of $u$.

For $k \ge m$, define $$R_k = \left\{u \in U : \mathcal I(u) \subseteq \left\{0, \frac{1}{k}, \dots, \frac{k-1}{k}, 1\right\}\right\}$$

For any $u \in R_k$, define 
$$a(u) = \#\left\{j \in \{0,\dots, k-1\} : \left[\frac{j}{k} \in \mathcal I(u)\right]\oplus\left[\frac{j+1}{k} \in \mathcal I(u)\right] \right\}$$
where $\oplus$ denotes exclusive-or. Note that $a(u)$ must be even and $a(u) \ge 2$.

For $k \ge m$, define 
$$
C_k = \{u \in R_k : a(u) = 2\}
$$

For any $u \in C_k$, define $\bar u: N \to \{0,1\}$ by
$$
\bar u(j) = \left\{
\begin{array}{ll}
1 & \textrm{if $u(j) > 0.5$} \\
0 & \textrm{if $u(j) \le 0.5$} \\
\end{array}
\right.
$$

Also define
$$
count(u) = \left|\{j \in M : \bar u(j)  = 1\}\right|
$$
and $rank(u,j)$ to be the rank of candidate $j$ according to $u$ for any $j \in M$, i.e.,
$$
rank(u,j) = \left|\{ j' \in M : u(j') \ge u(j)\}\right|
$$

For any $\mathbf u = (u_1,\dots u_n) \in (C_k)^n$
$$
g(\mathbf u) = \frac{\mathbf E[\sum_{i=1}^n u_i(J^*(\mathbf u))]}{\sum_{i=1}^n u_i(1)}
$$
$$
\bar g(\mathbf u) = \frac{\mathbf E[\sum_{i=1}^n \bar u_i(J^*(\mathbf u))]}{\sum_{i=1}^n \bar u_i(1)}
$$

For all $k \ge m$, define 
$$D^{(a)}_{k} = \left\{u \in C_{k} : (count(u) \le 2) \land (\bar u(1) =1)\right\}$$
$$D^{(b)}_{k} = \left\{u \in C_{k} : (count(u) = 1) \land (rank(u,1) > \lfloor m^{1/3} \rfloor)\right\}$$
$$D^{(c)}_{k} = \left\{u \in C_{k} : count(u) = rank(u,1) = \lfloor m^{1/3} \rfloor +1 \right\}$$
$$D_{k} = D^{(a)}_{k} \cup D^{(b)}_{k} \cup D^{(c)}_{k}$$

\section{A Good Truthful Voting Scheme}
\label{res1}

In this section we will prove that for $J^* = \frac{1}{2} J^{1,1} + \frac{1}{2} J^{1, \lfloor m^{1/3}\rfloor } \in \mathbf{Mech^T}$, we have $\operatorname{ratio}(J^*) = \Omega(m^{-2/3})$.

\begin{lemma}
\label{lemR}
$$\operatorname{ratio}(J^*) = \liminf_{k \to \infty} \ \min_{\mathbf u \in (R_k)^n}   \ \operatorname{ratio}(J^*, \mathbf u)$$
\end{lemma}

\begin{proof}
We have
$$
\operatorname{ratio}(J^*) = \inf_{\mathbf u \in V^n}   \ \operatorname{ratio}(J^*, \mathbf u)
$$

Therefore it is sufficient to show that for any $\mathbf u = (u_1,\dots, u_n)$ there exists an array of preference profiles 
$$\{\mathbf u^{(k)} = (\mathbf u^{(k)}_1, \dots, \mathbf u^{(k)}_n)\}_{k \ge 10m^{10}}$$
such that every $\mathbf u^{(k)} \in (R_k)^n$ and 
$$
\operatorname{ratio}(J^*, \mathbf u) = \lim_{k \to \infty} \operatorname{ratio}(J^*, \mathbf u^{(k)})
$$

We try to construct each $u^{(k)}_i$ as follow:

(1) For any $j_0, j_1 \in M$, let $u^{(k)}_i(j_0) < u^{(k)}_i(j_1) $  if and only if $$[u_i(j_0) < u_i(j_1)] \lor ([u_i(j_0) = u_i(j_1)] \land [j_0 < j_1])$$

(2) Minimize $$\sum_{j=1}^m  |u^{(k)}_i(j) - u_i(j)|$$

Obviously $\lim_{k \to \infty} u^{(k)}_i = u_i$.

Note that $J^* = 0.5J^{1,1} + 0.5J^{1, \lfloor m^{1/3}\rfloor}$ where both $J^{1,1}$ and $J^{1, \lfloor m^{1/3}\rfloor}$ break ties by numbering of candidates, which is exactly the same way we decide the ordering of candidates in $u^{(k)}_i$. Therefore we have $J^*(\mathbf u^{(k)}) = J^*(\mathbf u)$ for any $k \ge 10m^{10}$. Thus $f(\mathbf x) = \operatorname{ratio}(J^*, \mathbf x)$ can be seen as a continuous function for all possible profile $x = (x_1, ..., x_n) \in U^n$ such that for any $j_0, j_1 \in M$, $x_i(j_0) < x_i(j_1) $  if and only if $[u_i(j_0) < u_i(j_1)] \lor ([u_i(j_0) = u_i(j_1)] \land [j_0 < j_1])$.

Therefore by the property of continuous function we have 
$$
\operatorname{ratio}(J^*, \mathbf u) = \lim_{k \to \infty} \operatorname{ratio}(J^*, \mathbf u^{(k)})
$$

\end{proof}

\begin{lemma}
\label{lemG}
$$
\liminf_{k \to \infty} \ \min_{\mathbf u \in (R_k)^n}   \ \operatorname{ratio}(J^*, \mathbf u) = \liminf_{k \to \infty} \ \min_{\mathbf u \in (R_k)^n}   \ g(\mathbf u)
$$
\end{lemma}

\begin{proof}
The lemma can be written as
$$
\liminf_{k \to \infty} \ \min_{\mathbf u \in (R_k)^n}   \ \frac{\mathbf E[\sum_{i=1}^n u_i(J^*(\mathbf u))]}{\sum_{i=1}^n u_i(RV(\mathbf u))}= \liminf_{k \to \infty} \ \min_{\mathbf u \in (R_k)^n}   \ \frac{\mathbf E[\sum_{i=1}^n u_i(J^*(\mathbf u))]}{\sum_{i=1}^n u_i(1)}
$$

$(R_k)^n \subseteq U^n$. Futhermore, $J^*_U$ is neutral, i.e., when no tie is allowed in preference profile $J^*$ is neutral. Therefore if we constrain $u \in (R_k)^n$, $\frac{\mathbf E[\sum_{i=1}^n u_i(J^*(\mathbf u))]}{\sum_{i=1}^n u_i(RV(\mathbf u))}$ can be minimized when $RV(\mathbf u) = 1$.
\end{proof}

\begin{lemma} \cite{Rat2013}
\label{lemC}
$$
\liminf_{k \to \infty} \ \min_{\mathbf u \in (R_k)^n}   \ g(\mathbf u) = \liminf_{k \to \infty} \ \min_{\mathbf u \in (C_k)^n}   \ g(\mathbf u)
$$
\end{lemma}

\begin{proof}

We prove this lemma by induction on $\sum_{i=1}^n a(u_i)$ where $\mathbf u = (u_1,\dots, u_n)$. If $\sum_{i=1}^n a(u_i) = 2n$, $\mathbf u \in (C_k)^n$. Otherwise we can find a $i \in N$ such that $a(u_i) > 2$. Then we can find a interval $0 < j_0 \le j_1 < k$ such that $\{j_0/k, \dots, j_1/k\} \subseteq \mathcal I(u_i)$, $(j_0-1)/k \notin \mathcal I(u_i)$ and $(j_1+1)/k \notin \mathcal I(u_i)$. Then we can try to ``move'' the the interval left or right. Either moving it left or right can make both $g(\mathbf u)$ and $a(u_i)$ smaller. Thus by induction we can prove the lemma.

A more detailed proof can be found in Page 8 of \cite{Rat2013}. 

\end{proof}

\begin{lemma}
\label{lemD}
$$
\liminf_{k \to \infty} \ \min_{\mathbf u \in (C_k)^n}   \ g(\mathbf u) = \liminf_{k \to \infty} \ \min_{\mathbf u \in (D_k)^n}   \ g(\mathbf u) 
$$
\end{lemma}

\begin{proof}
From now on we set $k \to +\infty$. 

For any profile $\mathbf u = (u_1,\dots, u_n) \in (C_{k})^n \setminus (D_{k})^n$, there exists $i \in N$ such that $u_i \in C_{k} \setminus D_{k}$. If we can prove that there exists a $v_i \in D_{k}$ such that $g(v_i, u_{-i})\le g(u_i, u_{-i})$, the lemma is proved. Since when $k \to \infty$, $u(j) = \bar u(j)$ for any $u \in C_{k}$, it is sufficient to show that there exists $v_i \in D_{k}$ such that $\bar g(v_i, u_{-i})\le \bar g(u_i, u_{-i})$.

Since $\bar g(u) = \frac{\mathbf E[\sum_{i=1}^n \bar u_i(J^*(\mathbf u))]}{\sum_{i=1}^n \bar u_i(1)}$, to show $g(v_i, u_{-i}) \le g(u_i, u_{-i})$, it is sufficient to show that there exists a $v_i \in D_{k}$ such that all of the following conditions are satisfied:

(1) $J^*(v_i, u_{-i}) = J^*(u_i, u_{-i})$;

(2) $\bar v_i(1) \ge \bar u_i(1)$;

(3) $\bar v_i(j) \le \bar u_i(j)$ for any $j \ne 1$.

Also, notice that the mechanism $J^*$ can also be viewed as picking a random voter $i$, and then flip a coin. If it is a Head, choose voter $i$'s favorite candidate, otherwise choose a candidate among voter $i$'s $\lfloor m^{1/3} \rfloor$ most favorite candidates uniformly at random. Then the condition (1) can be rewrite as $$S(u_i) = S(v_i)$$ where for any $u \in C_{k}$, $$S(u) = \left(\{j \in M : r(u,j) = 1\}, \{j \in M : 1 < r(u,j) \le \lfloor m^{1/3} \rfloor\} \right)$$.

To prove this, we will do some discussion on the properties of $u_i$.

\textit{Case 1, $rank(u_i, 1) \le \lfloor m^{1/3}\rfloor$:}

\textit{Case 1.1, $rank(u_i, 1) = 1$ :} Let $v_i$ be the preference function that $\bar v_i(j) = 0$ for every $j \ne 1$ and $rank(v_i,j) = rank(u_i, j)$. Then $v_i \in D^{(a)}_{k}$ and all the three conditions are satisfied.

\textit{Case 1.2, $rank(u_i, 1) > 1$ :} Let $j'$ be the candidate with $rank(u_i, j') = 1$. Then let $v_i$ be the preference function such that $S(v_i) = S(u_i)$ and $\bar v_i(j) = 1$ if and only if $j \in \{1,j'\}$. Then $v_i \in D^{(a)}_{k}$ and all the three conditions are satisfied.

\textit{Case 2, $rank(u_i, 1) > \lfloor m^{1/3} \rfloor$:}

\textit{Case 2.1, $\bar u_i(1) = 1$ :} We have for any $j$ with $rank(u_i, j) \le \lfloor m^{1/3} \rfloor$, $\bar u_i(j) = 1$. Therefore we can let $v_i$ be the preference function such that $S(v_i) = S(u_i)$, $rank(v_i, 1) = \lfloor m^{1/3} \rfloor +1$ and $\bar v_i(j) = 1$ if and only if $j = 1$ or $rank(u_i, j) \le \lfloor m^{1/3} \rfloor$. Then $v_i \in D^{(c)}_{k}$ and all the three conditions are satisfied.

\textit{Case 2.2, $\bar u_i(1) = 0$ :} Let $j'$ be the candidate with $rank(u_i, j') = 1$. Then let $v_i$ be the preference function such that $S(v_i) = S(u_i)$ and $\bar v_i(j) = 1$ if and only if $j = j'$. Then $v_i \in D^{(b)}_{k}$ and all the three conditions are satisfied.

This is the end of the proof.

\end{proof}

\begin{corollary}
\label{cor2}
$$\operatorname{ratio}(J) = \liminf_{k \to \infty} \ \min_{\mathbf u \in (D_k)^n}   \ g(\mathbf u)$$
\end{corollary}

\begin{proof}
The corollary comes from combining Lemma \ref{lemR}, \ref{lemG}, \ref{lemC} and \ref{lemD}.
\end{proof}

\begin{theorem}
\label{Main1}
$\operatorname{ratio}(J^*) = \Omega(m^{-2/3})$ as $n \ge m$ and $m \to  \infty$
\end{theorem}

\begin{proof}
Follow Corollary \ref{cor2}, we have
$$
\operatorname{ratio}(J^*) = \liminf_{k \to \infty} \ \min_{\mathbf u \in (D_k)^n}   \ g(\mathbf u) = \liminf_{k \to \infty} \ \min_{\mathbf u \in (D_k)^n}   \ \bar g(\mathbf u)
$$

From now on consider a profile $\mathbf u \in (D_{k})^n$.

By anonymity of $J^*$, without loss of generality, we can assume that the first $a$ voters have valuation functions in $D^{(a)}_{k}$, the following $b$ voters' are in $D^{(b)}_{k}$ and the remaining $c$ voters' are in $D^{(c)}_{k}$.

Also, let random variable $T$ be the chosen voter. Let random variable $W$ be the winning candidate. Let $F$ be the event that we use the random-favorite mechanism, i.e., $J^{1,1}$. Let $X =\frac{\sum_{i=1}^n \bar u_i(W)}{\sum_{i=1}^n \bar u_i(1)}$. Then we can bound
\begin{eqnarray*}
\bar g(\mathbf u) &=& \mathbf E[X] \\
&\ge& \Pr[(T \le a) \land (\lnot F) \land (W = 1)] \\
&& + \Pr[(T \in \{a+1, \dots, a+b\}) \land F] \mathbf E[X | (T \in \{a+1, \dots, a+b\}) \land F] \\
&& + \Pr[(T \in \{a+b+1, \dots, n\}) \land \lnot F]\mathbf E[X | (T \in \{a+b+1, \dots, n\}) \land \lnot F]
\end{eqnarray*}

We have $$\Pr[(T \le a) \land (\lnot F) \land (W = 1)] = \frac{1}{2} \cdot \frac{a}{n} \cdot \frac{1}{\lfloor m^{1/3} \rfloor} = \frac{a}{2n\lfloor m^{1/3} \rfloor}$$.

We also have
\begin{eqnarray}
&&Pr[(T \in \{a+1, \dots, a+b\}) \land F] \mathbf E[X | (T \in \{a+1, \dots, a+b\}) \land F] \\
&=& \frac{1}{2}\cdot \frac{b}{n} \mathbf E[X | (T \in \{a+1, \dots, a+b\}) \land F] \\
&=& \frac{1}{2}\cdot \frac{b}{n} \cdot \frac{1}{b} \sum_{i=a+1}^{a+b} \mathbf E[X | (T=i) \land F] \\
&=& \frac{1}{2}\cdot \frac{b}{n} \cdot \frac{1}{b} \sum_{i=a+1}^{a+b} \frac{\sum_{j=1}^n \sum_{w: \bar u_i(w) = 1} \bar u_j(w)}{a+c} \\
&\ge& \frac{1}{2}\cdot \frac{b}{n}\cdot  \frac{1}{b} \sum_{i=a+1}^{a+b} \frac{\sum_{j=a+1}^{a+b} \sum_{w: \bar u_i(w) = 1} \bar u_j(w)}{a+c} \\
&=& \frac{1}{2}\cdot \frac{b}{n}\cdot  \frac{1}{b}\cdot \frac{1}{a+c}\sum_{w=1}^{m-1} |\{i \in \{a+1,\dots, a+b\} : \bar u_i(w) = 1\}|^2\\
&\ge& \frac{1}{2}\cdot \frac{b}{n}\cdot  \frac{1}{b}\cdot \frac{1}{a+c}(m-1)\left(\frac{b}{m-1}\right)^2 \\
&=& \frac{b^2}{2n (m-1) (a+c)}
\end{eqnarray}
where (4) comes from the fact that when $T = i \in \{a+1,\dots, a+b\}$ and we use $J^{1,1}$ mechanism, the winner is always the only $w$ with $\bar u_i(w) = 1$; equality (6) comes from recounting the terms; inequality (7) comes from the generalized mean inequality and the fact that $$\sum_{w=1}^{m-1} |\{i \in \{a+1,\dots, a+b\} : \bar u_i(w) = 1\}| = b$$.

Similarly we have
\begin{eqnarray*}
&&\Pr[(V \in \{a+b+1, \dots, n\}) \land \lnot F]\mathbf E[X | (V \in \{a+b+1, \dots, n\}) \land \lnot F] \\&\ge& \frac{c^2m^{1/3}}{2n(m-1)(a+c)}
\end{eqnarray*}

Therefore we have
$$
\bar g(\mathbf u) \ge \frac{a}{2nm^{1/3}} + \frac{b^2}{2n (m-1) (a+c)} + \frac{c^2m^{1/3}}{2n(m-1)(a+c)}
$$

Now we just consider the case that $n \ge m$ and $m \to +\infty$. Then we have
$$
\bar g(\mathbf u) = \Omega\left(\frac{a}{nm^{1/3}} + \frac{b^2}{n m (a+c)} + \frac{c^2m^{1/3}}{nm(a+c)}\right)
$$

Since $n = a + b + c$, either $a = \Omega(n)$, $b = \Omega(n)$ or $c = \Omega(n)$. If $a = \Omega(n)$, we have $\bar g(\mathbf u) \ge \frac{a}{nm^{1/3}} = \Omega(m^{-1/3})$ and we finish the prove. If $c = \Omega(n)$, we have $\bar g(\mathbf u) \ge \frac{c^2m^{1/3}}{nm(a+c)} = \Omega(m^{-2/3})$ and we finish. Then from now on we can assume that $b = \Omega(n)$ and $a, c = o(n)$. Then we have

$$
\bar g(\mathbf u) = \Omega\left(\frac{a}{nm^{1/3}} + \frac{n}{m(a+c)} + \frac{c^2m^{1/3}}{nm(a+c)}\right)
$$

If $a = \Omega(a+c)$, we have 
$$
\bar g(\mathbf u) = \Omega\left(\frac{a}{nm^{1/3}} + \frac{n}{ma}\right) = \Omega\left(\frac{a}{nm^{1/3}} + \frac{m^{-4/3}}{\frac{a}{nm^{1/3}}}\right) = \Omega(m^{-2/3})
$$
, which satisfies the theorem.

Now we can assume $c = \Omega(a+c)$, and we have
$$
\bar g(\mathbf u) = \Omega\left(\frac{n}{mc} + \frac{c}{nm^{-2/3}}\right) = \Omega\left(\frac{n}{mc} + \frac{mc}{n}m^{-4/3}\right) = \Omega(m^{-2/3})
$$

This is the end of the proof.
\end{proof}

\section{Discussion of the Preference Assumption}
\label{dis}

In the beginning of the paper we make the preference assumption requiring that every preference function $u$ must be a function from $M$ to $[0,1]$ such that there exists some $j_0$ and $j_1$ in $M$ with $u_i(j_0) = 0$ and $u_i(j_1) = 1$.

We know that this assumption is reasonable and sufficient to make the definition of relative social welfare well defined. However, in this section we will show that this assumption is necessary to make the problem interesting enough to study.

Now assume that a preference function can be any function from $M$ to $[0,1]$. Note that Lemma \ref{TO} still holds.

Let $n = m$ and we can construct profiles $\mathbf u^{(k)} = (u_1, \dots, u_n)$ as follow: 

(1) For any $i \in  \{1,\dots, n\}$, $u_i(i) > u_i(i+1) > \cdots > u_i(n) > u_i(1) > u_i(2) > \cdots > u_i(i-1)$.

(2) $u_k(k) = 1$ and $u_i(j) < \varepsilon$ for some extremely small $\varepsilon$ all $(i, j) \ne (k,k)$.

For any truthful mechanism $J$, we know that it is ordinal by Lemma \ref{TO}. Then no matter what $k$ is, the distribution of $J(\mathbf u^{(k)})$ remain unchanged. Then there exists some $i^* \in M$ such that $\Pr[J(\mathbf u^{(k)}) = i^*] \le 1/m$. 

Now we have 
$$
\operatorname{ratio}(J) \le \operatorname{ratio}(J, \mathbf u^{(i^*)}) \approx 1/m
$$
which not better than the trivial scheme that choose a random winner.

Thus if we don't add the preference assumption, the problem will become boring, which shows that the preference assumption is necessary.

\section{Conclusion}

In this paper we showed a model of cardinal voting which is proved to be reasonable, universal and as general as possible. Under such model, the paper provide a positive result providing a truthful mechanism with a relative social welfare of $\Omega(m^{-2/3})$, which is asymptotically tight.

In conclusion, this paper solves all the problems about relative social welfare of a truthful voting scheme with cardinal preference in asymptotic sense.

%
%

\begin{acknowledgements}
Thank Andrew Chi-Chih Yao for the discussion about the problem.
\end{acknowledgements}



\end{document}